%% file: main.tex
\documentclass[runningheads]{llncs}

\usepackage{multirow}
\usepackage[table,xcdraw]{xcolor}

\usepackage{float}
\usepackage{graphicx}
\usepackage{cite}
\usepackage{amsmath,amssymb,amsfonts}
\usepackage{subcaption}
\usepackage{hyperref}
\usepackage[noEnd]{algpseudocodex}
\usepackage{algorithm}
\usepackage[disable]{todonotes}
\usepackage{enumitem}
\usepackage{environ}
\usepackage{comment}
\usepackage{booktabs}
\usepackage{tabularx}
\usepackage{tabulary}
\usepackage{lipsum}
\usepackage{makecell}

\newfloat{algorithm}{t}{lop}

\makeatletter
\newcommand\newtag[2]{#1\def\@currentlabel{#1}\label{#2}}
\makeatother

\begin{document}
\title{A Lazy Abstraction Algorithm for Markov Decision Processes}
\subtitle{Theory and Initial Evaluation}

\author{D\'aniel Szekeres\inst{1}\orcidID{0000-0002-2912-028X} \and
Krist{\'o}f Marussy\inst{1}\orcidID{0000-0002-9135-8256} \and
Istv\'an Majzik\inst{1}\orcidID{0000-0002-1184-2882}}
\authorrunning{D. Szekeres \and I. Majzik}

\institute{Department of Measurement and Information Systems\\ Budapest University of Technology and Economics\\
\email{\{szekeres, marussy, majzik\}@mit.bme.hu}}
\maketitle 
\begin{abstract}
Analysis of Markov Decision Processes (MDP) is often hindered by state space explosion. Abstraction is a well-established technique in model checking to mitigate this issue. This paper presents a novel lazy abstraction method for MDP analysis based on adaptive simulation graphs. Refinement is performed only when new parts of the state space are explored, which makes partial exploration techniques like Bounded Real-Time Dynamic Programming (BRTDP) retain more merged states. Therefore, we propose a combination of lazy abstraction and BRTDP. To evaluate the performance of our algorithm, we conduct initial experiments using the Quantitative Verification Benchmark Set.

\keywords{Lazy abstraction \and Markov Decision Processes \and Abstraction refinement \and Probabilistic model checking}
\end{abstract}

\newcommand{\distr}[1]{\mathbb{D}(#1)}
\newcommand{\assignments}{\mathcal{A}}
\newcommand{\commandset}{\mathcal{C}}
\newcommand{\varset}{\mathcal{V}}
\newcommand{\range}[1]{\mathit{R}(#1)}
\newcommand{\dirac}[1]{\delta_{#1}}
\newcommand{\Unknown}{\mathit{Unknown}}
\newcommand{\True}{\mathit{True}}
\newcommand{\False}{\mathit{False}}
\newcommand{\dcomm}[1]{d_{#1}}
\newcommand{\gcomm}[1]{g_{#1}}
\newcommand{\ARG}{ASG}
\newcommand{\PARG}{P\ARG{}}
\newcommand{\ARGfull}{Adaptive Simulation Graph}
\newcommand{\PARGfull}{Probabilistic \ARGfull}
\newcommand{\targ}{T_\mathit{\PARG{}}}
\newcommand{\coveraction}{cover}

\newcommand{\repeattheorem}[1]{%
  \begingroup
  \renewcommand{\thetheorem}{\ref{#1}}%
  \expandafter\expandafter\expandafter\theorem
  \csname reptheorem@#1\endcsname
  \endtheorem
  \endgroup
}

\NewEnviron{reptheorem}[1]{%
  \global\expandafter\xdef\csname reptheorem@#1\endcsname{%
    \unexpanded\expandafter{\BODY}%
  }%
  \expandafter\theorem\BODY\unskip\label{#1}\endtheorem
}

\input{content/introduction}
    \input{content/relwork}
\input{content/background}

\input{content/theory_only_exact}
\input{content/evaluation}

\bibliographystyle{splncs04}
\bibliography{bibliography}

\input{content/appendix_exact_only}

\end{document}

%% file: content/introduction.tex
\section{Introduction}\label{sc:intro}
Ensuring the reliable operation of safety-critical systems, like railway interlocking systems and embedded controllers, is vital. Probabilistic model checking addresses this by offering an automated approach with formal mathematical guarantees for the analysis of quantitative properties, like reliability and availability \cite{KNP07a}. We focus on a fundamental task in probabilistic model checking: computing the worst-case probability of reaching an error state.

\emph{Markov Decision Processes (MDPs)} are discrete-time models able to describe both probabilistic and non-deterministic behavior, used in reliability and safety analysis for worst-case modeling of unknown factors. The analysis of other modeling formalisms, like Markov Automata or Probabilistic Timed Automata, can often be reduced to MDP analysis as well. 

State space explosion presents an obstacle for MDP model checking: as the number of components or variables increases, the state space may grow exponentially. Consequently, practical implementations face problems in representing the system in memory and the numerical solution methods also become intractable.

Abstraction aims to counteract this. Several abstraction-based techniques have been adapted to probabilistic systems, like CEGAR \cite{kattenbelt2009abstraction, hermanns2008probabilistic} and abstract interpretation \cite{CousotM12}. Partial state space exploration, like \emph{Bounded Real-Time Dynamic Programming (BRTDP)} \cite{Brazdil14Learning, kelmendi2018value} is another approach for counteracting it. As most existing MDP abstraction methods rely on computing the whole abstract model to choose a refinement, they do not lend themselves well to combination with partial state space exploration techniques. 

\emph{Lazy abstraction} \cite{henzinger2002lazy, mcmillan2006lazy} in contrast merges state-space exploration and refinement, making it a good candidate for this combination. However, no such method has been proposed for MDPs to our knowledge.

We adapt an existing lazy abstraction algorithm \cite{Toth20} to MDPs (\autoref{sc:lazy_mdp}). We combine it with BRTDP, benefiting from the synergy of lazy abstraction and partial state space exploration and enabling a trade-off between time and accuracy (\autoref{ssc:brtdp}). We evaluate the performance of the proposed algorithms using models from the Quantitative Verification Benchmark Set \cite{Hartmanns19QComp} (\autoref{sc:evaluation}). 

%% file: content/relwork.tex
\subsubsection{Related Work}\label{ssc:relwork}
\emph{Counterexample-Guided Abstraction Refinement (CEGAR)} \cite{Clarke00CEGAR} is a successful approach for abstraction-based model checking: it starts with a coarse abstraction and refines it based on abstract counterexamples.

Lazy abstraction, introduced in \cite{henzinger2002lazy}, improved CEGAR through \emph{on-demand refinement} during abstract state space exploration and varying precision from node to node in the state graph. An \emph{interpolant-based version} was proposed in \cite{mcmillan2006lazy}. This was adapted to timed automata \cite{herbreteau2013lazy}, introducing \emph{Adaptive Simulation Graphs (ASG)} as the abstract model. This allows \emph{earlier refinement}, cutting spurious paths before reaching a target, and a less expensive covering check. The ASG-based algorithm was adapted to explicit value abstraction of \emph{discrete variables} in \cite{Toth20}, which we, in turn, adapt to MDPs.

Different abstraction methods have been proposed for probabilistic systems. While some CEGAR-based methods employ MDPs as abstraction~\cite{dargenio02reduction, hermanns2008probabilistic, chadha2010counterexample}, others utilize stochastic games~\cite{parker2006game, wachter2010best, kattenbelt2009abstraction}, which we plan to incorporate in the future. Abstract interpretation has also been used for probabilistic systems~\cite{CousotM12, esparza2011probabilistic}. Some others include magnifying lens abstraction~\cite{de2007magnifying}, which explores the whole concrete state space but keeps only its subset in memory and assume-guarantee-style abstraction~\cite{komuravelli2012assume}, specialized for composite systems. To our knowledge, no lazy abstraction method has been proposed for probabilistic models yet.

The algorithm presented in this paper uses a symmetric representation constraint, resulting in an approach similar to \emph{bisimulation reduction} techniques \cite{dehnert13smt, kamaleson16finite}. The main difference is that until the whole \ARG{} is explored, only a limited version of ``bisimilarity'' holds which does not take the unexplored part of the state space into account, allowing coarser partitions on the already explored part. When combined with partial exploration, the algorithm can stop before exploring the full \ARG{}. \emph{Finite-horizon bisimulation minimization} \cite{kamaleson16finite} and \emph{incremental bisimulation abstraction-refinement} \cite{song2014incremental} are similar in that they employ a relaxed version of bisimulation. Both of them limit the bisimilarity to a fixed path length and compute exact quotients w.r.t the relaxed relation, while we base the relaxation on the currently explored state space and do not aim for computing the coarsest relation.

\emph{BRTDP} was introduced in \cite{mcmahan2005bounded} for Stochastic Shortest Paths, and \cite{Brazdil14Learning} applied it to \emph{general MDPs}. We combine it with our lazy abstraction algorithm.

%% file: content/background.tex
\section{Background and Notations}\label{sc:background}

$\distr{A}$ is the set of probability distributions over the set $A$. For $d \in \distr{A}, a \in A$, $d(a)$ denotes the probability measure of $a$ according to $d$. $f\colon A \hookrightarrow B$ means $f$ is a \emph{partial function} from $A$ to $B$, and $Supp(f)$ is the set of values for which $f$ is defined. For $d \in \distr{A}$, $Supp(d) = \{a \in A |  d(a) > 0\}$. $\dirac{x}$ is a Dirac distribution: $\dirac{x}(x)=1, \forall y\neq x\colon \dirac{x}(y)=0$.

\subsection{Markov Decision Process (MDP)} 

MDPs are low-level mathematical models that describe both probabilistic and non-deterministic behavior in discrete time. 

\begin{definition}[MDP]
   An MDP is a tuple $M = (S, Act, T, s_0)$, where $S$ is the set of states, $Act$ is the set of actions, $T\colon S \times Act \times S \xrightarrow{} [0, 1]$ is a probabilistic transition function s.t. $\forall s \in S, a \in Act\colon \sum_{s' \in S} T(s, a, s') \in \{0, 1\}$ and $s_0 \in S $ is the initial state.
\end{definition}

 An action $a \in Act$ is \emph{enabled} in $s \in S$ if $\sum_{s' \in S} T(s, a, s')=1$. In this case, $T(s, a) \in \distr{S}$ denotes the next state distribution after taking $a$ in $s$, defined as $T(s, a)(s') = T(s, a, s')$. The intuitive behavior of an MDP is as follows: starting in $s_0$ an action $a$ is chosen non-deterministically from those enabled in the current state $s_i$ in each step, and the next state is sampled from $T(s_i, a)$. A \emph{trace} of an MDP is an alternating list of states and actions $s_0 \xrightarrow{a_1} s_1 \xrightarrow{a_2} s_2 \xrightarrow{a_3} \dots $ such that $\forall i\colon T(s_{i-1}, a_i, s_i)>0$. Fixing a strategy for resolving the non-determinism, the set of traces can be equipped with a probability measure: intuitively, the probability of a trace is the product of the probability of landing in each state of the trace after taking the action specified by the strategy in the previous state. For a detailed formal treatment, see e.g. \cite{KNP07a}.

Given an MDP of the system behavior and a set of target (error) states $E$, we want to compute (an upper approximation of) the probability of the set of traces involving a state in $E$ with non-determinism resolved by a maximizing strategy: $\mathbb{P}_{max}(\{ (s_0 \xrightarrow{a_1} s_1 \xrightarrow{a_2} s_2 \xrightarrow{a_3} \dots)\ |\ \exists i \in \mathbb{Z}^+\colon s_i \in E \})$. The result is the same if we make all target states absorbing, allowing us to restrict the analysis to finite traces.

\subsubsection{Symbolic MDPs}
Most real-life models are specified symbolically using state variables and operations on them. We assume that the MDP is given by a set of variables $\varset$ and a set of probabilistic guarded commands $\commandset$. Each $v \in \varset$ has a set $\range{v}$ of values it can take, and an initial value $v_0 \in \range{v}$. A \emph{valuation} over $\varset$ is a function $val\colon \varset \xrightarrow{} \bigcup_{v \in \varset} \range{v}$ s.t. $\forall {v \in \varset}\colon val(v) \in \range{v}$, and $\mathit{VAL}_\varset$ is the set of all valuations over $\varset$. The initial valuation of the model is a valuation $val_0$ s.t. $\forall v \in \varset\colon val_0(v) = v_0$. The state space of this MDP is a subset of $\mathit{VAL}_\varset$, and its initial state is $val_0$. 

Let $\mathcal{B}_\varset$ denote the set of Boolean expressions over $\varset$, $\mathcal{E}_\varset^v$ the set of expressions over $\varset$ that result in an element of $\range{v}$, and $\mathcal{E}_\varset=\bigcup_{v \in \varset} \mathcal{E}_\varset^v$. An \emph{assignment} is a function $a\colon \varset \xrightarrow{} \mathcal{E}_\varset$, such that $\forall {v \in \varset}\colon a(v) \in \mathcal{E}_\varset^v$. Let $\assignments_\varset$ denote the set of assignments for $\varset$. $eval(e, val)$ for $e \in \mathcal{E}_\varset$ and $val \in \mathit{VAL}_\varset$ is the constant resulting from replacing each $v \in \varset$ in $e$ with $val(v)$. $eval(a, val)$ for $a \in \assignments_\varset$ is a valuation $val'$ such that $\forall {v \in \varset}\colon val'(v)=eval(a(v), val)$. $eval(d, val)$ for $d \in \distr{\assignments_\varset}$ is the distribution $d' \in \distr{\mathit{VAL}_\varset}$ such that $d'(val')=\sum_{a \in \{a \in Supp(d)\ |\ eval(a, val)=val'\} } d(a)$.

A \emph{command} $c \in \commandset$ consists of a guard $\gcomm{c} \in \mathcal{B}_\varset$ and a result distribution over assignments $\dcomm{c} \in \distr{\assignments_\varset}$. $c$ is \emph{enabled by} $val$ iff $eval(\gcomm{c}, val)=\True$. Let $a^c_i \in Supp(\dcomm{c})$ denote the $i$th assignment of command $c$ for a fixed ordering. The enabled actions in each state $val$ of the represented MDP are the commands enabled by $val$, and taking the command $c$ results in the distribution $eval(\dcomm{c}, val)$. Widely used MDP description formats, like that of \textsc{PRISM}~\cite{KNP11} or the \textsc{JANI}~\cite{Budde17Jani} format can be mapped to this low-level description.

Our running example is given by the following variables and commands:
\begin{align*}
    & \varset = \{x, y\}, \range{x}=\range{y}=\mathbb{N}, x_0=y_0=0 \\
    & \mathbf{c_1}\colon [\mathit{true}] \ 0.8: (x':=x+1 \wedge y':=y), 0.2: (x':=x \wedge y':=y) \\
    & \mathbf{c_2}\colon [x==0]\ 1.0: (x':=1 \wedge y':=2) \\
    & \mathbf{c_3}\colon [x==2 \wedge y==2]\ 1.0: (x':=x \wedge y':=3)
\end{align*}
$c_1$ is enabled in every state, and it increments $x$ by $1$ with probability $0.8$. $c_2$ is enabled when $x=0$, and always sets $y$ to $2$ and $x$ to $1$. $c_3$ is enabled when $x$ is $2$ and $y$ is 2, and sets $y$ to $3$. \autoref{fig:mdp_example} shows this MDP.

\subsection{Lazy Abstraction} 
Abstraction-refinement methods mitigate state-space explosion by disregarding information present in the original \emph{concrete} model to create an \emph{abstract} model that is iteratively refined until a conclusion is reached. Lazy abstraction performs refinement on-the-fly and only on a subset of the state space.

For checking safety properties in the qualitative case, a conservative abstraction overapproximates the \emph{reachable state set}. In the probabilistic setting, the \emph{probability} of reaching a target state in the abstract model overapproximates that in the concrete one, which we will prove for the proposed algorithm. 

We build on the lazy abstraction method of \cite{Toth20} for non-probabilistic systems. It constructs an \emph{\ARGfull{} (\ARG{})} with nodes labeled by both a concrete and an abstract state: the concrete state represents all states in the abstract state regarding possible action sequences. The abstract state labels start very coarse and are refined as needed. \emph{Covering edges} indicate that action sequences starting from the coverer node encompass those starting from the covered node, eliminating the need to explore paths from the covered node. 

If an action is enabled in at least one concrete state described by the abstract label of a node, but not in the concrete label, the abstract label is \emph{strengthened} by removing states with the action enabled. This operation can trigger additional strengthenings. The algorithm terminates once all enabled actions in non-covered nodes have been explored. The abstract labels in the finished \ARG{} cover all reachable concrete states, and contain a target state only if one is reachable.
\subsubsection{Abstract domains}

Abstract states are described using an \emph{abstract domain}. For a set of concrete states $S$, an abstract domain $D=( \hat{S}, \preceq, \alpha, \gamma )$ consists of the abstract state set $\hat{S}$, a partial ordering $\preceq \ \subseteq \hat{S} \times \hat{S}$, an abstraction function $\alpha\colon 2^S \xrightarrow{} \hat{S}$ and a concretization function $\gamma\colon \hat{S} \xrightarrow{} 2^S$ satisfying
$\forall {A \in 2^S, \hat{a} \in \hat{S}}\colon \alpha(A) \preceq \hat{a} \iff A \subseteq \gamma(\hat{a})$. 
$\gamma$ lets us treat abstract states as sets of concrete states; 
we write ``$s \in \hat{s}$'' for $s \in \gamma(\hat{s})$ when $\gamma$ is clear from the context. $x \preceq y$ denotes $(x, y) \in \ \preceq$. $\hat{S}$ has two special elements: $\top$ and $\bot$ satisfying $\gamma(\top)=S, \gamma(\bot)=\{\}$.

Our lazy abstraction algorithms are domain agnostic, but need an abstract domain for $S=\mathit{VAL}_\varset$ with the following operations.

For $a \in \assignments_\varset$ and $\hat{s} \in \hat{S}$, \emph{abstract post operator} $eval(a,  \hat{s}) \in \hat{S}$ applies an assignment in the abstract state space: $eval(a,  \hat{s}) = \alpha(\{eval(a, s) | s \in \hat{s}\})$. For $b \in \mathcal{B}_\varset$, $eval(b, \hat{s}) \in \{\True, \False, \Unknown\}$ denotes evaluating $b$ in the abstract state space: $\True$ if $b$ evaluates to $\True$ for \emph{all} $s \in \hat{s}$, $\False$ if $b$ evaluates to $\False$ for \emph{all} $s \in \hat{s}$, otherwise $\Unknown$.

We also need a \emph{block} operation: for an abstract state $\hat{s} \in \hat{S}$, a Boolean expression $b \in \mathcal{B}_\varset$ and a concrete state $ s \in \hat{s}$ s.t. $eval(b, s)=\False$, $\hat{s}' = block(\hat{s}, b, s)$ is an abstract state s.t. $\hat{s}' \preceq \hat{s}, s \in \hat{s}', eval(b, \hat{s}')=\False$. Its goal is to give a new abstract state by removing at least those states from $\hat{s}$ which satisfy $b$ (potentially others as well) while keeping $s$. 

The abstract states must be representable as Boolean expressions: for each $\hat{s} \in \hat{S}$ a $b_{\hat{s}} \in \mathcal{B}_\varset$ must exist s.t. $\forall s \in S\colon eval(b_{\hat{s}}, s)=\True \iff s \in \hat{s}$. Relying on this, we will freely use abstract states in place of Boolean expressions. 

We will use the \emph{explicit value domain} $D_{expl}$ (abstract states correspond to tracking only a subset of $\varset$) as an example throughout the paper which we implemented in our prototype, along with predicate abstraction $D_{pred}$ (abstract states are Boolean predicates over $\varset$). A partial function $pval\colon \varset \hookrightarrow \bigcup_{v \in \varset} \range{v}$ s.t. $\forall v \in Supp(pval)\colon pval(v) \in \range{v}$ is called a \emph{partial valuation}, $\mathit{PVAL}_\varset$ denotes the set of all partial valuations over $\varset$. Description of these domains and their operations can be seen in \autoref{tab:domains}, assuming a concrete state set $\mathit{VAL}_\varset$. 

The lazy abstraction algorithm does not use the abstraction and concretization functions $\alpha$ and $\gamma$ and the abstract post operator $eval(a,  \hat{s})$ directly, only as arguments of a block operation (see later), so they need not be efficiently computable if the corresponding block operation can be implemented efficiently.

\begin{table}[h]
    \setlength{\tabcolsep}{6pt}
    \centering
    \begin{tabulary}{0.93\linewidth}{p{0.11\linewidth}CC}
    \toprule
          & \textbf{EXPL} & \textbf{PRED}  \\
    \midrule
    $\hat{S}$ & $\mathit{PVAL}_\varset \cup \{\bot_{expl}\}$ & $\mathcal{B}_\mathcal{V}$ \\\midrule
     $\preceq$ & $pval \preceq pval' \iff (Supp(pval') \subseteq Supp(pval) \wedge \forall{v \in Supp(pval')}\colon pval(v)=pval'(v) )$ & $b_1 \preceq b_2 \iff (b_1 \implies b_2)$ \\\midrule
    $\alpha$ & \(Supp(\alpha(A)) = \{ v \in \varset\ |\ \exists k \in \range{v}\colon \forall val \in A\colon val(v)=k \}\) and $\forall v \in Supp(\alpha(A))\colon \alpha(A)(v)=k \iff \forall val \in A\colon val(v)=k$ & $\bigvee_{\mathit{val}\in A}(\forall v \in \mathcal{V}\colon v=\mathit{val}(v))$ \\\midrule
    $\gamma$ & $\gamma(pval)=\{val \in \mathit{VAL}_\varset\ |\ \forall v \in Supp(pval)\colon val(v)=pval(v)\}$ & $\gamma(b)=\{val \in \mathit{VAL}_\varset\ |\ \mathit{val}\vdash b )\}$ \\\midrule
    $\top, \bot$ & $\top$ is the empty valuation, $\bot$ is a non-valuation element $\bot_{expl}$ representing contradiction & $\top=\True, \bot=\False$ \\\midrule
    Boolean representation & $b_{pval}=(\bigwedge_{v \in Supp(pval)} v=pval(v))$ & identity (already a Boolean expression) \\\midrule
   $eval(b, \hat{s})$  & Substituting the values in $Supp(\hat{{s}})$, and deciding the satisfiability of the result. & $\True$ if $\hat{s} \implies b$, $\False$ if $\hat{s} \implies \neg b$, $\Unknown$ if neither. \\\midrule
   $eval(a, \hat{s})$ & Substituting known variables into the result expressions. If this results in a constant, that is the result, else unknown. & strongest postcondition
   \\
    \bottomrule
    \end{tabulary}
    \caption{Properties and operations of the explicit value and predicate domains}
    \label{tab:domains}
\end{table}

%% file: content/theory_only_exact.tex
\section{Lazy Abstraction for MDPs}\label{sc:lazy_mdp}

\begin{figure}[t]
    \centering
    \begin{subfigure}[t]{0.49\textwidth}
        \includegraphics[width=\textwidth]{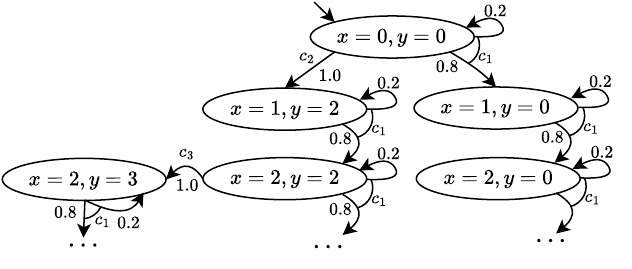}
        \caption{}
        \label{fig:mdp_example}
    \end{subfigure}
    \begin{subfigure}[t]{0.49\textwidth}
        \includegraphics[width=\textwidth]{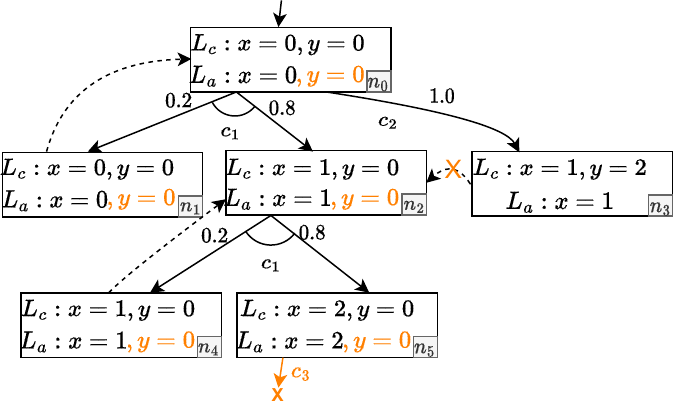}
        \caption{}
        \label{fig:asg_example}
    \end{subfigure}
    \caption{Our running example MDP (a) and an in-progress \PARG{} for it (b). The result of a refinement step is marked with orange}
\end{figure}

\todo{improve intuitive presentation and example explanations if possible, especially refinement strategy}
\todo{highlight the advantages of ASG-based lazy}

\input{content/codes}

Now we adapt the lazy algorithm to symbolic MDPs given by a variable set $\varset$ and a command set $\commandset_0$. Given a target formula $\phi \in \mathcal{E}_\varset$, the goal is to compute the maximal probability of reaching a state $s$ s.t. $eval(\phi, s)=\True$. 

We select abstract domain $(\hat{S}, \preceq, \alpha, \gamma)$. The set of commands is extended with a target command: $\commandset=\commandset_0 \cup \{(\phi, \dirac{id})$\}, where $id$ is the identity assignment. A node is a target if this command is enabled in it. This ensures that the finished \ARG{} 
contains a node labeled with a target state exactly if a target state is reachable in the concrete state space \cite{Toth20}. 

\subsubsection{Abstract model}
We use a probabilistic extension of the \ARG. A direct adaptation of the non-probabilistic lazy algorithm by switching to probabilistic actions would overapproximate the target probability with no control over the approximation. Therefore, we use a stricter, symmetric representation constraint for the relation between the concrete and abstract labels of a Probabilistic \ARG{} node. 

\begin{definition}[\PARG{}]
A \PARGfull{} is a tuple $(N, E_T, E_C, L_c, L_a)$, where $N$ is a set of nodes, $E_T \subseteq N \times \commandset \times \distr{N}$ is a set of transition ``edges'' from nodes to node distributions labeled with commands, $E_C \subseteq N \times N$ is a set of directed \emph{covering} edges, $L_c: N \xrightarrow{} \mathit{VAL}_\varset$ is the \emph{concrete labeling function}, $L_a: N \xrightarrow{} \hat{S}$ is the \emph{abstract labeling function}. 
\end{definition}

For an edge $e=(n, c, d) \in E_T$, $n^e_i$ is the $i$th element of $Supp(d)$ for a fixed ordering. A \PARG{} is \emph{well-labeled}, if it satisfies the constraints in \autoref{tab:constraints}. The main difference from the original \ARG{} is that in \ref{b2)}, we use a distribution of results instead of a single one, and \ref{a2)} requires the set of enabled actions to be exactly the same in all concrete states contained in the abstract label. 

In the original lazy algorithm, refinement is performed when an action \textit{disabled} in $L_c(n)$, but \textit{enabled} in some element of the abstract label $L_a(n)$, and we refine by blocking the guard of the action from $L_a(n)$. Here, we \textit{also} refine when an action is \textit{enabled} in $L_c(n)$ but \textit{disabled} in some element of $L_a(n)$ by blocking the \textit{negation} of the guard. We also need to adapt soundness to probabilities, see \autoref{thm:bi_soundness}. \ref{d)} is a technical constraint to make our proofs easier. 

\begin{table}[h]
    \centering
    \setlength{\tabcolsep}{6pt}
    \begin{tabulary}{\linewidth}{@{}JL@{}}
    \toprule
     \textbf{Constraint} & \textbf{Formalisation} \\
     \midrule
      \newtag{\textbf{A1})}{a1)} Abstract label contains the concrete label: & $\forall n \in N: L_c(n) \in L_a(n)$ \\\addlinespace
     
     \newtag{\textbf{A2)}}{a2)} Concrete label exactly represents the whole abstract label with respect to the enabled commands & $\forall n \in N: \forall c \in \commandset: eval(\gcomm{c}, L_c(n)) = eval(\gcomm{c}, L_a(n))$ \\\addlinespace
    
    \newtag{\textbf{B1)}}{b1)} The command of a transition edge is enabled in the concrete label of the source & $\forall (n, c, \cdot) \in E_T \colon{} \discretionary{}{}{} \mathit{eval}(\gcomm{c}, L_c(n))=\True$ \\\addlinespace
     
    \newtag{\textbf{B2)}}{b2)} For transition edges $e=(n, c, d) \in E_T$, the $i$th result node is consistent with the $i$th assignment:
        same probability, concrete label is the result of the assignment, abstract label overapproximates the result & $d(n^e_i)=\dcomm{c}(a^c_i)$ $L_c(n^e_i)=eval(a^c_i, L_c(n))$ $eval(a^c_i, L_a(n)) \preceq L_a(n^e_i)$ \\\addlinespace

    \newtag{\textbf{C1)}}{c1)} Abstract label of covering node contains the concrete label of covered node & $\forall (n, n') \in E_C\colon \discretionary{}{}{} L_c(n) \in L_a(n')$ \\\addlinespace
    
    \newtag{\textbf{C2)}}{c2)} Covering node is at least as abstract as the covered node & $\forall (n, n') \in E_C\colon \discretionary{}{}{} L_a(n) \preceq L_a(n')$ \\\addlinespace
    
    \newtag{\textbf{C3)}}{c3)} Covering node is not covered & $\forall (n, n') \in E_C\colon \discretionary{}{}{} \neg \exists (n', n'') \in E_C $ \\\addlinespace

    \newtag{\textbf{D1)}}{d)} At most one node labeled with a given concrete label can be non-covered & $\forall n, n' \in N\colon n \neq n' \wedge L_c(n)=L_c(n') \implies (\exists n''\colon (n, n'') \in E_C \vee (n',n'') \in E_C) $
     \\\bottomrule
    \end{tabulary}
    \caption{\PARG{} constraints}
    \label{tab:constraints}
\end{table}

\begin{example}
    \autoref{fig:asg_example} shows an example \PARG{} with black (the orange part is a refinement example explained later). The abstract label tracks only $x$ in all nodes (this could differ from node to node in general). $L_c$ is contained in $L_a$ for all nodes as the value of $x$ is the same in $L_c$ and $L_a$ (\ref{a1)}). 

    $n_0$ covers $n_1$, $n_2$ covers $n_4$ and $n_3$, satisfying \ref{c1)} and \ref{c2)}. $n_3$ could cover $n_4$ according to the labels, but it would violate \ref{c3)}. 

    This \PARG{} is unfinished, $n_4$ and $n_5$ are not expanded. $n_0$ is an example for the remaining constraints. \ref{a2)} is satisfied, as tracking $x$ in $L_a$ is enough to disable $c_3$, and both $c_1$ and $c_2$ are enabled in $L_c$ and everywhere in $L_a$. As the outgoing edges are labeled with $c_1$ and $c_2$, \ref{b1)} is satisfied. Let $e=(n_0, c_1, d)$ be the $c_1$ edge from $n_0$. The assignments in $c_1$ are $a^{c_1}_1=(x':=x+1 \wedge y':=y)$, paired with $n^e_1=n_2$ and $a^{c_1}_2=(x':=x \wedge y':=y)$ paired with $n^e_2=n_1$, which satisfies \ref{b2)}. E.g. for $a^{c_1}_1$: $eval(a^{c_1}_1, L_c(n_0))=eval((x':=x+1 \wedge y':=y), (x=0, y=0))=(x=1, y=0)$, $eval(a^{c_1}_1, L_a(n_0))=(x=1) \preceq (x=1)$ and $\dcomm{c_1}(a^{c_1}_1)=0.8=d(n_2)$.
\end{example} 

\subsubsection{Exploration} 
\autoref{alg:construction} shows \PARG{} construction. An initial node $n_0$ labeled $L_c(n_0)=val_0, L_a(n_0)=\top$ is extended to a well-labeled \PARG{} with each node either covered or expanded. 
\autoref{alg:block} shows blocking an expression from $L_a(n)$, used during refinement.

When removing a node $n$ from the waitlist, we check whether $\exists n_c \neq n \in N\colon L_c(n') \in L_a(n_c)$ s.t. $n_c$ is not covered.  If so, a covering edge $(n', n_c)$ is created and $L_a(n')$ is strengthened for \ref{c2)} to hold. Else, it is expanded.

If a node $n \in N$ is selected for expansion, we check for each $c \in \commandset$ whether $eval(\gcomm{c}, L_c(n))=\True$. If so, a new node $n'_i$ is created for each $a_i \in Supp(\dcomm{c})$ with $L_c(n'_i)=eval(a_i, L_c(n)), L_a(n'_i)=\top$, and a transition edge $(n, c, d_e)$ is created such that $d_e(n'_i)=\dcomm{c}(a_i)$ for $i=1\dots |Supp(\dcomm{c})|$. Because of \ref{a2)}, if the abstract label contains states where the transition is disabled, we remove them by blocking out the negated guard (Line \ref{ln:must_strengthening}). 

If $\mathit{eval(\gcomm{c}, L_c(n))=\False}$, we compute $eval(\gcomm{c}, L_a(n))$. If $\False$, we move on to the next command. If $\Unknown$ (cannot be $\True$, as $L_c \in L_a$), \ref{a2)} is violated, so $L_a(n)$ needs to be strengthened: a new abstract label is computed as $\hat{s}'=block(L_a(n), \gcomm{c}, L_c(n))$. Because of the contract of $block$, $eval(\gcomm{c}, \hat{s}')=\False$, so this command no longer causes a constraint violation (Line \ref{ln:strengthen_for_command}).

\subsubsection{Refinement} 
Refinement is interleaved with exploring the abstract state space. Whenever the $L_a(n)$ changes for some $n \in N$, the constraints may be violated. If constraint \ref{c1)} is violated, the problematic covering edge is removed from $E_C$. This makes $n$ non-covered, so we expand it later (Line \ref{ln:remove_cover}). 

If constraint \ref{c2)} is violated by covering edge $(n, n') \in E_C$, but constraint \ref{c1)} still holds, the current $L_a(n)$ must be replaced with $\hat{s}'$ such that $L_c(n) \in \hat{s}'$, $\hat{s}' \preceq L_a(n')$ and $\hat{s}' \preceq L_a(n)$ (referring to the current $L_a$). An appropriate $\hat{s}'$ is $block(L_a(n), \neg L_a(n), L_c(n))$ (Line \ref{ln:strengthen_for_cover}). 

Assume that \ref{b2)} is violated by an edge $e=(n, c, d)$. Because of how the \PARG{} is constructed, the concrete label and probability subconstraints of \ref{b2)} must still hold, but the abstract label part is violated by some node $n^e_i$: $L_a$ of $n^e_i$ no longer overapproximates applying $a^c_i$ (the assignment that led to its creation) to $L_a$ of its parent. The violation caused by this assignment is eliminated by changing $L_a(n)$ to $block(L_a(n), eval^{-1}(a^c_i, L_a(n^e_i), L_c(n))$ (Line \ref{ln:strengthen_parent}).

Strengthenings may create new violations, but all of them are eliminated after finite steps (if the concrete label can be finitely represented in the abstract domain), and we continue expanding non-covered nodes. Efficient implementations of the algorithm can employ sequence interpolation to strengthen the whole path up to the root at once \cite{Toth20}, which we do when using the predicate domain, as we observed that both simple weakest-precondition-based refinement and binary interpolation lead to predicates growing very fast.

\begin{example} 
\autoref{fig:asg_example} shows an example of refinement in orange. Expanding $n_5$, we realize $c_3$ is enabled in some states described by abstract label $x=2$, but not in the concrete label $x=2, y=0$. Thus, we strengthen $n_5$ by blocking the guard $x==2 \wedge y==2$, resulting in the abstract label $x=2, y=0$. 

This triggers another strengthening, as $L_a(n_5)$ no longer overapproximates applying $x'=x+1 \wedge y'=y$ to $L_a(n_2)$. $n_2$ is also strengthened, removing a cover edge as $L_c(n_3)$ is no longer contained in $L_a(n_2)$. 
Strengthening a covering node also strengthens the covered nodes if the covering remains (see $n_4$ and $n_1$).
\end{example}

\subsubsection{Numerical analysis} The finished \PARG{} can be treated as an MDP $(N, \mathcal{C} \cup \coveraction, \targ, n_0)$. Regarding $\targ$, for a non-covered node $n$, a command $c \in C$ is enabled if there is an edge $(n, c, d) \in E_T$ in the \PARG{}, and $\targ(n, c)=d$. The only action in covered nodes is $\coveraction$, which results in their covering node.

\begin{reptheorem}{thm:bi_soundness}
    The maximal/minimal probability of reaching a target node in the \PARG{} of an MDP $M$ is the same as in $M$.
\end{reptheorem}

Refer to the Appendix for proofs.

 Due to the symmetry of constraint \ref{a2)}, this abstraction can be considered similar to bisimulation-reduction, but not aiming for the coarsest bisimulation. Constructing the \PARG{} can be computationally cheaper than the coarsest bisimulation, but the larger state space may result in more expensive numerical computation. The advantages appear when the abstraction is combined with partial state space exploration, as most bisimulation reduction algorithms in the literature cannot be done on the fly. Comparing our method to bisimulation reduction in depth (both theoretically and empirically) is planned for future work. 

\subsection{Combining with BRTDP}\label{ssc:brtdp}

Bounded Real-Time Dynamic Programming (BRTDP) \cite{mcmahan2005bounded, Brazdil14Learning} approximates the value function of an MDP \emph{during state space exploration}. It maintains an upper and a lower bound ($U$ and $L$) by generating traces and updating the bounds for the encountered states. In each step, the optimal action is chosen according to the current $U$. The strategy for choosing a state from the result distribution is a parameter of the algorithm, for which we implemented the RANDOM and DIFF\_BASED trace generation strategies from \cite{Brazdil14Learning}.

Our lazy abstraction algorithm combines well with such methods. As refinement is performed during expansion, the abstract states in an in-progress \PARG{} are coarser than those in a finished \PARG{}. Thus, if BRTDP reaches the required threshold before constructing the full \PARG{}, the abstract labels remain coarser. Existing probabilistic CEGAR schemes like \cite{kattenbelt2009abstraction} cannot benefit from this, as they need lower and upper value bounds for all nodes for refinement, while BRTDP works best when only the value of the initial node is needed. This combination enables a controlled trade-off between time and accuracy.

The skeleton of the algorithm is the same as the BRTDP algorithm described in \cite{Brazdil14Learning} for general MDPs. The difference is that instead of generating traces from the final \PARG{}, we use the steps of building the \PARG{} for trace generation. As such, refinement is also possible during trace simulation, potentially removing covers used in previous simulations.

\begin{reptheorem}{thm:bi_brtdp_soundness}
    The maximal probability of reaching a target state is always between $L(n_0)$ and $U(n_0)$ if BRTDP is used with \PARG{} construction steps.
\end{reptheorem}

This theorem is non-trivial, as the traces are not generated for the \emph{finished} \PARG{}, but for an in-progress version where transient cover edges can exist which would not be present when finished. The main idea behind its proof is that if a value is propagated through a cover edge, then the value has been updated only based on traces consistent with the cover edge (as we would have already removed it if any trace inconsistent with it had been explored).

%% file: content/codes.tex
\begin{algorithm}
	\caption{\PARG{} construction}\label{alg:construction}
	\begin{algorithmic}[1]
        \State $N \gets \{n_0\}; E_T \gets \{\}; E_C \gets \{\}$
        ; $L_c(n_0) \gets s_0; L_a(n_0) \gets \top$
        ; waitlist $\gets \{n_0\}$
        \While{waitlist is not empty}
            \State $n \in$ waitlist
            ; waitlist $\gets$ waitlist$\setminus \{n\}$
            \If{$\exists\ n_c \neq n \in N: L_c(n) \in L_a(n_c) \wedge n_c \text{ not covered}$} 
                \State $E_C \gets E_C \cup \{ (n, n_c) \}$
                ; $Block(n, \neg L_a(n_c))$
            \Else{} \Comment{Expansion}
            \ForAll{$c \in \commandset$}
                \If{$eval(\gcomm{c}, L_c(n))=\True$} \Comment{$c$ is enabled in $L_c(n)$}
                    \If{$c$ is target command} mark $n$ as target \EndIf
                    \If{$eval(\gcomm{c}, L_a(n))=\Unknown$} \label{ln:must_strengthening}
                    \State Block($n$, $\neg g_c$) \Comment{$c$ can be disabled in $L_a$, so we refine}
                    \EndIf
                    \ForAll{$a_i \in Supp(\dcomm{c})$}
                        \State $N \gets N \cup \{n_{new}\}$;
                        $L_c(n_{new}) \gets eval(a, L_c(n))$;
                        $L_a(n_{new}) \gets \top$;
                        $\delta(n_{new}) \gets p^c_i$
                        waitlist $\gets$ waitlist $\cup\ n_{new}$ 
                    \EndFor
                    \State $E_T \gets E_T \cup (n, c, \delta)$ \label{ln:new_edge}
                \ElsIf{$eval(g(c), L_a(n))=\Unknown$} \label{ln:strengthen_for_command} 
                    \State Block($n$, $\gcomm{c}$) \Comment{$c$ can be enabled in $L_a$, but is not in $L_c$, so we refine}
                \EndIf
            \EndFor
            \EndIf
        \EndWhile
        \State \Return \PARG{} $(N, E_T, E_C, L_c, L_a)$
        \algstore{alg}
	\end{algorithmic}
\end{algorithm}

\begin{algorithm}
	\caption{Block(n, $\phi$)}\label{alg:block}
	\begin{algorithmic}[1]
        \algrestore{alg}
        \Require $eval(\phi, L_c(n))=\False$
        \State $L_a(n) \gets block(L_a(n), \phi, L_c(n))$
        \ForAll{$(n', n) \in E_C$} 
            \Comment{Check nodes $n'$ covered by $n$}
            \If{$L_c(n') \notin L_a(n)$} \Comment{Remove if new $L_a$ cannot cover}
            \State $E_C \gets E_C \setminus (n', n); waitlist \gets waitlist \cup\ n'$ \label{ln:remove_cover}
            \Else{} Block($n'$, $\neg L_a(n)$) \label{ln:strengthen_for_cover} 
            \Comment{Else refine covered node}
            \EndIf
        \EndFor
        \State Let $e=(n_{pre}, c, d) \in E_T\ s.t.\ n \in Supp(d)$
        \Comment{$n_{pre}$ denotes the parent of this node}
        \State Let $a^c_i$ s.t. $n^e_i=n$
        \Comment{$a^c_i$ is the assignment which resulted in the node $n$}
        \State Block($n_{pre}, \neg eval^{-1}(a^c_i, L_a(n))$) \label{ln:strengthen_parent}
        \Comment{Making sure that \ref{b2)} still holds for $L_a$}
	\end{algorithmic}
\end{algorithm}

%% file: content/evaluation.tex
\section{Evaluation}\label{sc:evaluation}

\begin{figure}
    \centering
    \includegraphics[width=\linewidth]{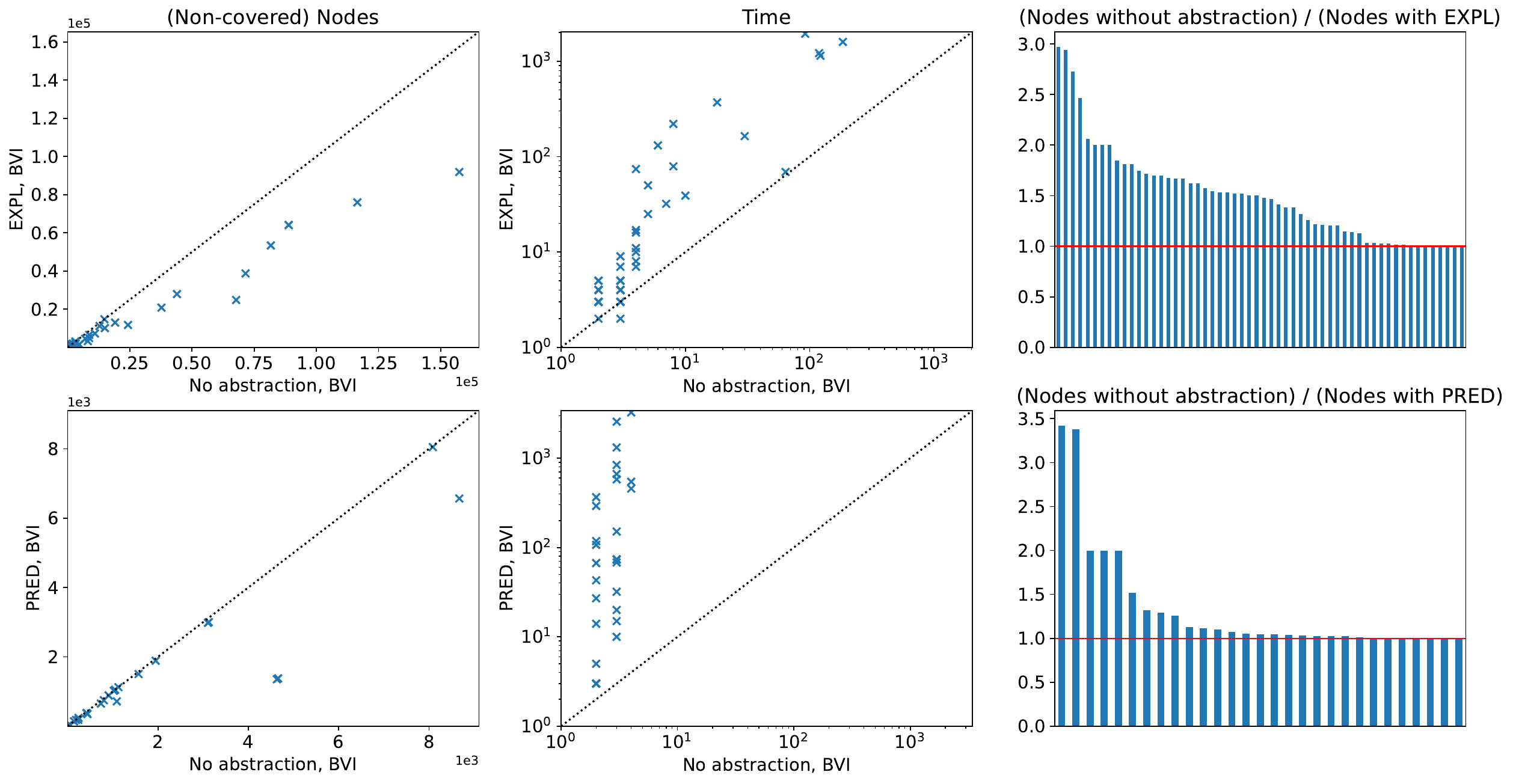}
    \caption{Comparison of standard and abstract BVI with the EXPL (top row) and PRED (bottom row) domain. Columns: 1\textsuperscript{st}: number of (non-covered) nodes, 2\textsuperscript{nd}: running time on log scale, 3\textsuperscript{rd}: ratio of the original and abstract state space size (red line marks 1, which would be no reduction). Only for inputs where the two algorithms compared in the plot terminated. 
    }
    \label{fig:bvi_plots}
\end{figure}

\noindent\textbf{Implementation} Our prototype\footnote{\url{https://github.com/szdan97/probabilistic-theta/tree/prob-proto}}
with the explicit value and predicate domains is implemented in the Theta model checker \cite{theta-fmcad2017}, taking JANI models\cite{Budde17Jani} as input. Only properties of the form $P_{max}(p\ U\ q)=?$ are supported where $p, q$ are Boolean constraints. Action result probabilities must be constant. The locations of JANI models are always tracked in the abstract states and covering can only occur between nodes with the same locations. To fairly assess the algorithms rather than the implementations, we implemented both Bounded Value Iteration (BVI) \cite{Baier17Ensuring} and BRTDP as MDP solution techniques in Theta both with and without lazy abstraction. We precompute almost sure reachability and avoidance to speed up the numerical solution if BVI is used. 

The main metrics of interest are state space size and running time. As the numerical computations mostly scale with the non-covered \PARG{} nodes, we are especially interested in the number of non-covered nodes. Our evaluation focuses on the following research questions:

\noindent\emph{\textbf{RQ1.} How does lazy abstraction affect the state space size and analysis time?}

\noindent\emph{\textbf{RQ2.} Does the combination of lazy abstraction with BRTDP lead to reduced abstract model size when converged? How does it affect the running time?}

\subsubsection{Setup} We used the 104 MDP model-property pairs from the Quantitative Verification Benchmark Set \cite{Hartmanns19QComp} compatible with our current implementation. The experiments were conducted using BenchExec\cite{Beyer19Reliable}, running them on the \emph{Komondor HPC} \footnote{https://hpc.kifu.hu/en} with AMD EPYC\textsuperscript{TM} 7763 CPUs, each run getting 8 CPU cores, 16GB RAM, and a 1-hour timeout. Convergence threshold 
was $10^{-6}$ (absolute) for all algorithms.

\subsubsection{Results and discussion}

\noindent\textbf{RQ1.}
\autoref{fig:bvi_plots} shows the BVI results.
Lazy abstraction often significantly reduced the state space: for example, a 3-fold reduction was possible for \textit{beb.3-4 [N=3, prop: GaveUp]} (from 4632 nodes to 1559 non-covered nodes) and \textit{csma.2-6 [prop: all\_before\_max]} (from 67741 to 24837) with EXPL, and for \textit{beb.3-4 [N=3, prop: GaveUp]} (from 4632 to 1354) with PRED. There are also inputs where the explicit domain could not reduce the state space size (e.g. \textit{blocksworld.5, cdrive.2}, but that is expected because of its low granularity.

Measuring the analysis time, it turned out that the overhead of more complex operations during exploration outweighed the benefits of numerical computations on a smaller state space. The overhead is apparent for EXPL, but it is much less severe than for PRED. 

Predicate abstraction was sometimes able to reduce the state space more than the explicit domain when it terminated before timeout, but the opposite was also present (related plots can be found in the appendix). The overhead of interpolating using an SMT solver was often too large, and so PRED often failed to terminate in time. 

We identified several abstraction-specific optimization possibilities for the implementation. For one, the interpolants returned by Z3 were often very large and had a redundant structure, which we could mitigate through structural simplification, improving both time and memory efficiency. Investigating alternative refinement methods and other solvers could also lead to better results.

Another opportunity for optimization is that like the non-probabilistic version \cite{Toth20}, we currently create multiple nodes with the same $L_C$ if a concrete state is reachable on multiple paths. As one of them always covers the others, only that is explored. However, according to our investigation, there are inputs where this led to multiple magnitudes of increase in the number of nodes during exploration. Merging such nodes instead would solve this issue, but refinement must be slightly changed as the path up to the root will no longer be unambiguous. We already had some preliminary measurements with promising results regarding this modification. 

\begin{figure}
    \centering
    \includegraphics[width=\linewidth]{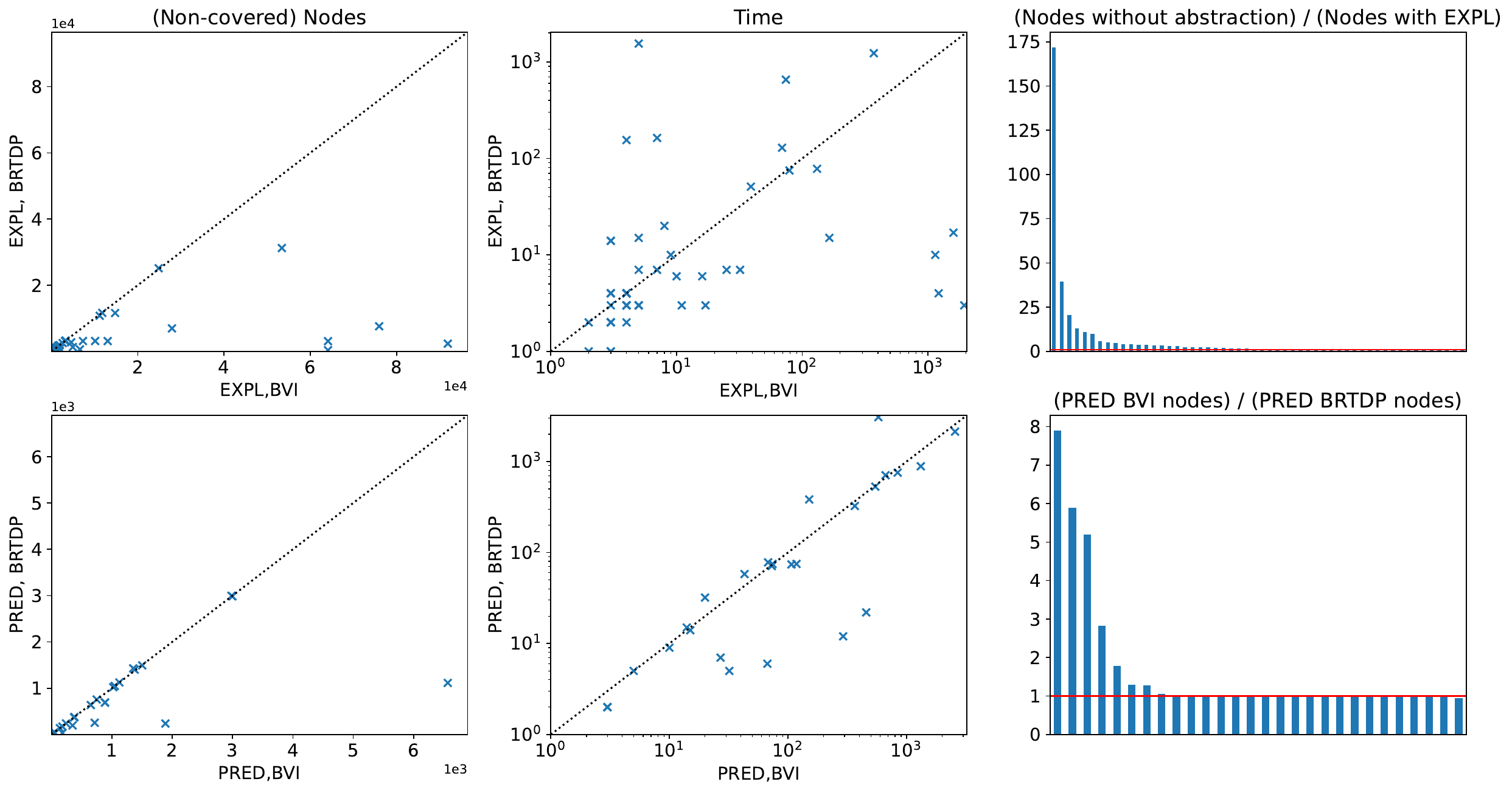}
    \caption{Comparison of abstract BVI and BRTDP. Same plot types as \protect\autoref{fig:bvi_plots}.}
    \label{fig:brtdp_vs_bvi_plots}
\end{figure}

\noindent\textbf{RQ2.} The simulation-based nature of BRTDP makes it harder to gauge the benefits of abstract BRTDP compared to standard BRTDP: there were inputs where abstract BRTDP converged with fewer nodes and where the concrete one did. (We relegated plots related to this comparison to the appendix.)

The benefits of abstract BRTDP compared to abstract BVI are much more apparent. The plots in \autoref{fig:brtdp_vs_bvi_plots} show this comparison (the results of the strategy leading to fewer nodes were used for each input-domain pair for BRTDP).

The highest relative state space size reduction was on \textit{zeroconf [N=20, k=2, !reset, prop: correct]} (170-fold from 64109 to 373 non-covered nodes) for EXPL and \textit{pnueli-zuck.3 [prop: live]} (from 1888 to 239 non-covered nodes) for PRED. There were inputs where no further reduction was achieved though (e.g. \textit{blocksworld.5, cdrive.2, rectangle-tireworld.5, ij.10} for both domains). When both BVI and BRTDP terminated, BRTDP was often able to do so in less time, especially with PRED (the results are much more two-sided for EXPL). 

\section{Conclusions}\label{sc:conclusions}

We proposed a lazy abstraction algorithm for symbolic MDPs and combined it with BRTDP. We provided numerical evaluation for different versions of the proposed algorithm using the Quantitative Verification Benchmark Set, comparing them to explicitly computing the concrete state space.

The initial experimental evaluation shows potential in the proposed algorithm, especially if reducing the state space is paramount for staying within the memory limits. As the time overhead introduced by more complex computations in the state space exploration often outweighed the gains from analyzing a smaller state space, we plan to explore possible improvements for this aspect. Further measurements using other benchmark sets and more parameterizations of the scalable models in the QVBS are also planned.

Additionally, we wish to explore \emph{fine-grained transitioning} between the strict representation version (see \ref{a2)}) and a more direct overapproximating adaptation by incorporating ideas from \emph{game-based abstraction refinement}, moving the algorithm closer to standard abstraction approaches instead of bisimulation reduction.

%% file: content/appendix_exact_only.tex
\newpage
\appendix

\section{Proofs}

For a \PARG{} trace $n_0 \xrightarrow{act_1} n_1 \xrightarrow{act_2} n_2 \xrightarrow{act_3} \dots \xrightarrow{act_k} n_k$ such that $act_k \neq \coveraction$, its \emph{coverless representation} is constructed such that for each $i$ where $act_i=\coveraction$, $act_i$ and $n_i$ is dropped from the trace. The original trace can always be reconstructed from the coverless representation, as 

\begin{lemma}\label{lm:abstract_enabledness}
    For any well-labeled finished \PARG{} $P$ for the MDP $M$, let $s_0 \xrightarrow{c_1} s_1 \xrightarrow{c_2} s_2 \xrightarrow{c_3} \dots \xrightarrow{c_k} s_k$ be a trace of $M$. Then there must exist a trace of $P$ with coverless representation $n_0 \xrightarrow{act_1} n_1 \xrightarrow{act_2} n_2 \xrightarrow{act_3} \dots \xrightarrow{act_k} n_k$ such that for all $i=1..k: s_i \in L_a(n_i)$.
\end{lemma}
\begin{proof}
    For any non-covered node $n$ and concrete state $s$, if $s \in L_a(n)$, then all commands enabled in $s$ are enabled in $n$ because of \ref{a2)} and \ref{b2)} (the concrete label has at least those commands enabled that are enabled in any state described by the abstract label, and the actions enabled in a non-covered node are the commands enabled in its concrete label). Because of \ref{c2)}, if there is a covering edge $(n, n')$ and $s \in L_a(n)$, the $s \in L_a(n')$ is also true. Because of the abstract label subconstraint of \ref{b2)}, if $s \in L_a(n)$ and choosing the assignment in $n$ leads to $n'$, then $eval(a, s) \in L_a(n')$. 

    Based on these observations, the theorem can be proven by induction on the trace length. It holds for the 0-length trace $s_0$, as $s_0 = L_c(n_0) \in L_a(n_0)$. Then for an $M$ trace with length $i+1$, we know that the length $i$ prefix has a corresponding $P$ trace ending in $n_i$. If this is not a covered node, then $c_{i+1}$ is enabled in it as (from the induction assumption) $s_i \in L_a(n_i)$, then we choose any assignment of $c_{i+1}$ that leads to $s_{i+1}$, and choose the same assignment in $P$ leading to a node $n_{i+1}$. As we observed, because of the abstract label subconstraint of \ref{b2)}, $s_{i+1} \in L_a(n_{i+q})$. If $n_i$ is covered, we take the $\coveraction$ action to $n'$, where $c_{i+1}$ is enabled, as $s_i \in L_a(n')$ because of \ref{c2)}. We take $c_{i+1}$ and an appropriate assignment to $n_{i+1}$, and $s_{i+1} \in n_{i+1}$ is true for the same reason as in the other case. This concludes our induction proof.
\end{proof}

A consequence of this is that if a node $n$ is covered by $n'$ in a finished \PARG{}, then for all traces from $L_c(n)$, there exists a trace from $L_c(n')$ with the same commands and assignments chosen.

\begin{lemma}\label{thm:upper_soundness}
    The maximal probability of reaching a target node in a \PARG{} for an MDP $M$ is at least as high as the maximal probability of reaching a target state in $M$.
\end{lemma}

\begin{proof}
Fix a maximizing memoryless strategy $\sigma$ for the MDP $M$ (we can restrict ourselves to memoryless strategies as an optimal memoryless strategy for an MDP is also an optimal general strategy if the optimization goal is a reachability probability). This induces a Discrete-Time Markov Chain $M^{\sigma}$, where the probability of reaching a target state from the initial state is the maximal probability of reaching a target in $M$. We will construct a strategy $\hat{\sigma}: N* \xrightarrow{} \commandset \cup \{\coveraction\} $ for the \PARG{} $P$ that results in at least as high probability for reaching a target node as $\sigma$. Note that $\hat{\sigma}$ is \emph{not} a memoryless strategy.

The aim of $\hat{\sigma}$ will be to copy the traces of $M^\sigma$. The non-trivial part of this is that because of covers, multiple states of $M^\sigma$ can correspond to the same node in the \PARG{} -- that is why $\hat{\sigma}$ is not a memoryless strategy but uses the whole trace to select an action. 

$\hat{\sigma}$ will obviously always choose $\coveraction$ if the trace ends in a covered node, as there is no other enabled action. Let $n_0 \xrightarrow{c_1} n_1 \xrightarrow{c_2} n_2 \xrightarrow{c_3} \dots \xrightarrow{c_n} c_n$ be a coverless representation of a \PARG{} trace $\hat{\tau}$. Because of \ref{b2)}, there is a one-to-one correspondence between the assignments of $c_i$ and the result nodes of the transition edge. Let $a_i$ be the assignment corresponding to $n_i$ in this trace.

We have two cases:
\begin{enumerate}
    \item A concrete trace $\tau=s_0 \xrightarrow{c_1} s_1 \xrightarrow{c_2} s_2$ can be constructed from $\hat{\tau}$ by starting in $s_0$, choosing $c_i$ and applying $a_i$ to $s_{i-1}$ in the $i$th step. This way, we can compute which concrete state we would actually be in if this assignment sequence happened in the original state space. We will call $\tau$ the concretization of $\hat{\tau}$. This is possible only if $c_i$ is enabled in $c_{i-1}$, which need not always be true -- that will be the second case. Let $\hat{\sigma}(\hat{\tau})$ be the same command as $\sigma(\text{last}(\tau))$.
    \item The trace is not concretizable. In this case, $\hat{\sigma}(\hat{\tau})$ can be anything, as such traces will never be generated by $\hat{\sigma}$ if it is constructed according to the previous case for concretizable traces, so it cannot change the induced MDP $P^{\hat{\sigma}}$.
\end{enumerate}

This is a valid strategy, as \autoref{lm:abstract_enabledness} ensures the selected command is enabled in the \PARG{} node. The fact that the choice for non-concretizable traces does not matter can be seen by induction: 
\begin{itemize}
    \item The only possible 1-length trace $n_1$ is trivially concretizable, the commands enabled in $n_0$ are exactly those enabled in $L_c(n_0)=s_0$.
    \item Let $\tau$ be the concretization of a concretizable trace $\hat{\tau}$. At the end of $\hat{\tau}$, $\hat{\sigma}$ chooses a command that is enabled at the end of $\tau$, resulting in a trace one step longer that is still concretizable.
\end{itemize}

Although there is always a unique assignment that results in entering a node in the \PARG{}, this is not true in general in the original MDP (e.g. in the state $x=1$, the assignments $x':=x+1$ and $x':=2$ both lead to the state $x=2$ -- if these are the assignments of the same command, we do not know which one was chosen), so there are multiple $P^{\hat{\sigma}}$ traces with the same concretization. 

Assume for a moment that concretization creates a one-to-one correspondence between abstract and concrete traces, so the aforementioned problem is not present. In this case, the probability subconstraint of \ref{b2)} would ensure that the probability of $\hat{\tau}$ according to the strategy $\hat{\sigma}$ is the same as the probability of its concretization $\tau$ according to $\sigma$. 

This can be generalized to the case when there are multiple abstract traces $\hat{\tau}_1, \hat{\tau}_2, \dots \hat{\tau}_k$ with the same concretization $\tau$: in this case, $\mathbb{P}_{\sigma}(\tau)=\sum_{i=1}^k \mathbb{P}_{\hat{\sigma}}(\hat{\tau}_i)$. 

Proof sketch for proving this by induction: if at the first command of a trace there are two assignments $a_1$ and $a_2$ in a command leading to the same concrete state $s'$ when applied to $s$, then the probability of going from $s$ to $s'$ is $\mathbb{P}(a_1)+\mathbb{P}(a_2)$, and the probability of the whole trace is $\mathbb{P}(a_1)+\mathbb{P}(a_2)$ times the probability of the suffix $\tau'$; in the \PARG{}, these two assignments result in different nodes, but the abstract traces corresponding to the suffix are available from both nodes (or there covering nodes, which only multiplies the probability of the suffix by $1.0$), so assuming the sum of the probability of these abstract traces is $\mathbb{P}(\tau')$, then the probability of the whole trace is $\mathbb{P}(a_1)\mathbb{P}(\tau')+\mathbb{P}(a_2)\mathbb{P}(\tau')=(\mathbb{P}(a_1)+\mathbb{P}(a_2))\mathbb{P}(\tau')$, which is the original probability.

Let $abstr(\tau)$ denote the set of abstract traces that have the concretization $\tau$, and $\mathbb{P}(abstr(\tau))$ is the sum of their probabilities. From the previous statement, we have $\mathbb{P}(\tau) = \mathbb{P}(abstr(\tau))$  Let $concr(\hat{\tau})$ be the concretization of $\hat{\tau}$. Observe that concretization is deterministic, so for any $\tau_1, \tau_2$, $abstr(\tau_1) \cap abstr(\tau_2) = \emptyset$. Let $\hat{T}$ be the set of $P^{\hat{\sigma}}$ traces that lead to target nodes, and $\hat{T}_c=(\hat{\tau} \in \hat{T} | concr(\hat{\tau}) \text{ is a target trace in $M^\sigma$} )$. $\hat{T}_c$ need not be the whole $\hat{T}$, as target nodes can cover non-target nodes, let $\hat{T}_s$ denote the "spurious" target traces $\hat{T}\setminus \hat{T}_c$. Let $T$ denote the set of target traces of $M^\sigma$. Observe that $\hat{T}_c = \bigcup_{\tau \in T} abstr(\tau) $. Now we can prove that the target probability on $P^{\hat{\sigma}}$ is at least as high as in $M^\sigma$:
\begin{align*}
    \mathbb{P}(\hat{T})&=\mathbb{P}(\hat{T}_c)+\mathbb{P}(\hat{T}_s)=(\sum_{\hat{\tau}\in\hat{T}_c} \mathbb{P}(\hat{\tau}))+\mathbb{P}(\hat{T}_s)=\\
    &=(\sum_{\tau\in T} \mathbb{P}(abstr(\tau)))+\mathbb{P}(\hat{T}_s)=(\sum_{\tau\in T} \mathbb{P}(\tau))+\mathbb{P}(\hat{T}_s)=\mathbb{P}(T)+\mathbb{P}(\hat{T}_s)
\end{align*}

As $\mathbb{P}(\hat{T}_s)$ is non-negative, and $\mathbb{P}(T)$ is exactly the target probability in $M^\sigma$, we have proven that the target probability in $\mathbb{P}^{\hat{\sigma}}$ is at least as high as in $M^\sigma$. $\hat{\sigma}$ is not necessarily a maximizing strategy, so the maximal probability in $P$ can be even higher. $\sigma$ was chosen to be maximizing, so this proves the theorem. 

The probability of reaching a target state in $M^\sigma$ is the sum of the probabilities of all traces reaching a target state. For each of these traces, we have a ``copy'' in the MDP induced by $\hat{\sigma}$ for \PARG{}, \emph{which has the same probability}, as the same commands are chosen throughout the trace and the probability of choosing a given assignment in a step of the trace is the same in the \PARG{} as in $M^\sigma$. There are other traces as well in the \PARG{} with strategy $\hat{\sigma}$ that end in target nodes: non-concretizable traces because of covering nodes enabling commands that are not enabled in the covered node and traces that end in target nodes that cover non-target nodes.

As $\hat{\sigma}$ already results in at least as high target reachability probability as an optimal strategy on $M$, and it might not even be optimal, this shows that the maximal probability in the \PARG{} must be at least as high as in $M$.

\end{proof}

\begin{lemma}\label{thm:lower_soundness}
The maximal probability of reaching a target node in a \PARG{} for an MDP $M$ is at most as high as the maximal probability of reaching a target state in $M$.
\end{lemma}

\begin{proof}

\emph{Sketch:
}
The proof is done similarly to that of \autoref{thm:upper_soundness}, just the other way around. The main difference is that strategies on $M$ have less information in the trace than strategies in the \PARG{}, as nodes in the \PARG{} are basically labeled with which assignment was the last. This leads to a problem when two assignments of a command can lead to the same state in $M$, as a strategy on the \PARG{} may choose a different action in them (e.g. they can be covered by different nodes). 

Because of this, we construct a different MDP $M'$ from the symbolic description of $M$, where the states also contain the assignment that was applied last time. This MDP is bisimilar to $M$, the bisimulation relation being based on forgetting the last assignment -- neither the ``targetness'' nor the enabled commands and their result distribution depends on the last assignment if the current state is known, so it is easy to prove that this is a bisimulation. As such, the maximal target probability on $M'$ is the same as on $M$.  

For any maximizing memoryless strategy $\hat{\sigma}$ on the \PARG{}, we can construct a non-memoryless strategy $\sigma$ on $M'$ such that $\sigma(\tau)=\hat{\sigma}(last(abstr(\tau)))$, which is now unique because of the assignment labels. The strategy is valid, which can be proven using an ``inverted'' \autoref{lm:abstract_enabledness} based on \ref{a2)}. This results in at least as high a probability for reaching a target state as in the \PARG{}. 

\end{proof}

\repeattheorem{thm:bi_soundness}
\begin{proof}
The theorem is a result of combining \autoref{thm:upper_soundness} and \ref{thm:lower_soundness}: the probability in the \PARG{} is both a lower and an upper approximation of the probability in the MDP, so it must be the same.
\end{proof}

\begin{lemma}\label{thm:upper_brtdp_soundness}
        The maximal probability of reaching a target state is always below $U(n_0)$ if BRTDP is used with \PARG{} construction steps.
\end{lemma}
\begin{proof}
    \emph{Sketch:}
    Let $M'$ be an assignment-labeled version of $M$ similarly to the proof of \autoref{thm:lower_soundness}. For an element $s$ and an assignment $a$, the $M'$ state $s \oplus a$ corresponds to the $M$ state $s$ with the latest assignment being $a$. The initial state is a special $s_0 \oplus \bot$ element without an assignment label.
    For all already existing \PARG{} nodes $n$, $U(n) \geq V(L_c(n) \oplus a_n)$, where $V$ is the real value function of $M'$, $a_n$ is the assignment corresponding to $n$ according to constraint \ref{b2)}, and $L_c(n) \oplus a_n$ is a state of $M'$. This can be proven by induction. This is of course true when the node is created, as it is initialized to $1.0$ unless it is a non-target absorbing node, in which case $U(n)=V(L_c(n))=0.0$. Whenever $U(n)$ is updated in later steps, we have two cases:

    \begin{enumerate}
        \item $n$ is non-covered, so it is updated according to a transition edge. In this case, the update corresponds to an update on $M'$ according to an upper approximation of the value function $U'$, which is set to $U'(s)=U(m)$ if a node exists with $L_c(m)\oplus a_m$ and $1.0$ otherwise. Because of the induction hypothesis, this is a valid upper bound for the value function, which means that computing a Bellman update using it is a valid
        \item $n$ is covered by $n'$. We do not yet know whether the currently covering node will remain covering until the finished \PARG{}, so we cannot say that all action sequences possible from $L_c(n) \oplus a_n$ are possible from $L_c(n')\oplus a_(n')$. However, as the covering edge still exists, we know that those action sequences that are possible from $L_c(n)\oplus a_n$ but not from $L_c(n')\oplus a_{n'}$ have not been explored yet: along the way of these sequences, there exists a non-covered non-expanded node. The value of this node is currently overapproximated by $1.0$. So for all action sequences possible from $L_c(n)\oplus a_n$, they are either also possible from $L_c(n')\oplus a_{n'}$, or there is a prefix of it that is possible from it, and ends in a node whose value approximation is set to $1.0$. Because of this (and the fact that we are aiming to \emph{maximize} the value), the current value approximation of $n'$ must be higher than the value of $L_c(n)\oplus a_n$, so copying it to $n$ does not violate the induction assumption.
    \end{enumerate}

    As $M'$ is bisimilar to $M$, $V(s_0)=V(s_0 \oplus \bot)$, so $U(n_0)$ is also an upper bound for $s_0$.
    
\end{proof}

\begin{lemma}\label{thm:lower_brtdp_soundness}
        The maximal probability of reaching a target state is always above $L(n_0)$ if BRTDP is used with \PARG{} construction steps.
\end{lemma}
\begin{proof}
    \emph{Sketch:}
    This can be proven similarly to \autoref{thm:upper_brtdp_soundness}. The upper approximation $U$ is replaced in all statements with the lower approximation $L$. The non-covered case goes the same.

    If $n$ is covered by $n'$ in the finished \PARG{}, we know that all action sequences starting from $L_c(n')\oplus a_{n'}$ are also available from $L_c(n) \oplus a_n$. When an update happens during BRTDP, we do not know this yet. However, for all action sequences available from $L_c(n')\oplus a_{n'}$ and not available from $L_c(n) \oplus a_n$, as the covering edge is still there, there must be a state whose current lower approximation is $0$, as the corresponding \PARG{} node has not been expanded yet. So the current $L(n')$ must be lower than $V(L_c(n)\oplus a_n)$.
\end{proof}

\repeattheorem{thm:bi_brtdp_soundness}
\begin{proof}
    Consequence of combining \autoref{thm:upper_brtdp_soundness} and \autoref{thm:lower_brtdp_soundness}.
\end{proof}

\section{Further Experiment Results}

\autoref{fig:bvi_domain_plots} compares the PRED and EXPL domains when using BVI. Although EXPL always wins in running time, PRED can sometimes reduce the number of nodes further. If the overhead of refinement in the PRED domain can be mitigated making it terminate in time on more inputs, the better reduction capability might be observed on even more inputs.  

\autoref{fig:brtdp_plots} compares the results using BRTDP both for the baseline and abstraction-based analysis. No clear advantage of either of them can be observed. 

\begin{figure}[h]
    \centering
    \includegraphics[width=\linewidth]{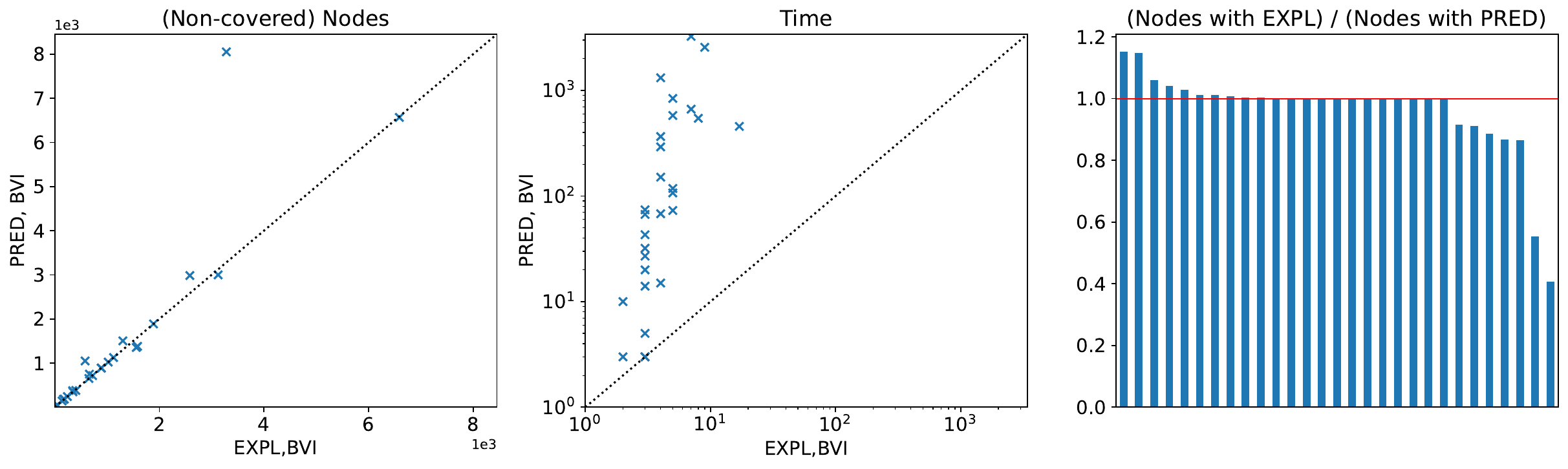}
    \caption{Comparison of EXPL and PRED BVI.}
    \label{fig:bvi_domain_plots}
\end{figure}
\begin{figure}[h]
    \centering
    \includegraphics[width=\linewidth]{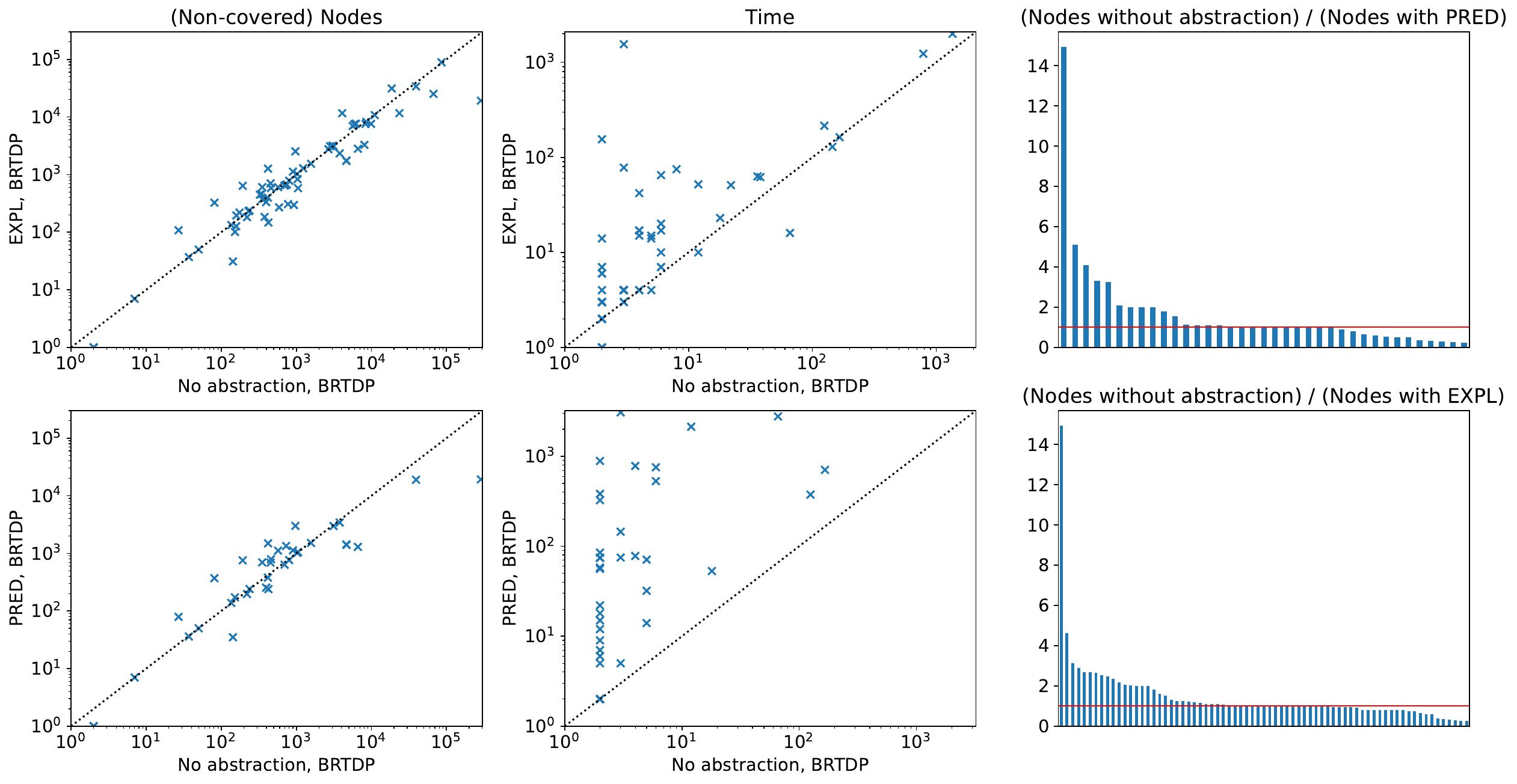}
    \caption{Comparison of abstract and standard BRTDP. Unlike in the other two figures, the node number comparison uses log scale, as linear scaling made it hard to see.}
    \label{fig:brtdp_plots}
\end{figure}